\documentclass[a4paper]{article}
\usepackage{amsmath}

\newtheorem{theorem}{Theorem}

\newtheorem{corollary}{Corollary}
\newtheorem{definition}{Definition}

\newtheorem{proposition}{Proposition}

\newtheorem{example}{Example}
\newenvironment{proof}[1][Proof]{\noindent\textbf{#1.} }{\ \rule{0.5em}{0.5em}}

\begin{document}

\title{\textbf{Cyclic Codes over Some Finite Quaternion Integer Rings}}
\author{Mehmet \"{O}zen, Murat G\"{u}zeltepe \\
{\small Department of Mathematics, Sakarya University, TR54187 Sakarya,
Turkey}}
\date{}
\maketitle

\begin{abstract}
In this paper, cyclic codes are established over some finite
quaternion integer rings with respect to the quaternion Mannheim
distance, and decoding algorithm for these codes is given.
\end{abstract}


\bigskip \textsl{2000 AMS Classification:}{\small \ 94B05, 94B15, 94B35, 94B60}


\section{Introduction }

Mannheim distance, which is much better suited for coding over two
dimensional signal space than the Hamming distance, was introduced
by Huber \cite{1}. Moreover, Huber constructed one Mannheim error
correcting codes, which are suitable for quadrature amplitude
modulation (QAM)-type modulations \cite{1}. Cyclic codes over some
finite rings with respect to the Mannheim metric were obtained by
using Gaussian integers in \cite{2}. Later, in \cite{3}, using
quaternion Mannheim metric, also called Lipschitz metric \cite{4},
perfect codes over some finite quaternion integer rings were
obtained and these codes were decoded.

The rest of this paper is organized as follows. In Section II,
quaternion integers and some fundamental algebraic concepts have
been considered. In Section III, we construct cyclic codes over
some quaternion integer rings with respect to quaternion Mannheim
metric.

\section{Quaternion Integers}

\begin{definition}The Hamilton Quaternion Algebra over the set of the real numbers
($\mathcal{R} $), denoted by $H(\mathcal{R})$, is the associative
unital algebra given by the following representation:

i)$H(\mathcal{R})$ is the free $\mathcal{R}$ module over the
symbols $1,i,j,k$, that is, $ H(R) = \{ a_0  + a_1 i +$ $ a_2 j +
a_3 k:\;a_0 ,a_1 ,a_2 ,a_3  \in R\}$;

ii)1 is the multiplicative unit;

iii) $ i^2  = j^2  = k^2  =  - 1$;

iv) $ ij =  - ji = k,\;ik =  - ki = j,\;jk =  - kj = i$ \cite{5}.
\end{definition}

The set $H(\mathcal{Z})$, $ \;H(\mathcal{Z}) = \left\{ {a_0  + a_1
i + a_2 j + a_3 k:\;a_0 ,a_1 ,a_2 ,a_3  \in \mathcal{Z}}
\right\}$, is a subset of $H(\mathcal{R})$, where $\mathcal{Z}$ is
the set of all integers. If $ q = a_0  + a_1 i + a_2 j + a_3 k$ is
a quaternion integer, its conjugate quaternion is $\overline q =
a_0  - (a_1 i + a_2 j + a_3 k)$.The norm of  $q$ is $ N(q) =
q.\overline q  = a_0^2  + a_1^2  + a_2^2  + a_3^2$. A quaternion
integer consists of two parts which are the complete part and the
vector part. Let $ q = a_0  + a_1 i + a_2 j + a_3 k $ be a
quaternion integer. Then its complete part is $a_0$ and its vector
part is $ a_1 i + a_2 j + a_3 k $. The commutative property of
multiplication does not hold for quaternion integers. However, if
the vector parts of quaternion integers are parallel to each
other, then their product is commutative. Define $ H(K_1 ) $ as
follows:$$ \;H(K_1 ) = \left\{ {a_0  + a_1 (i + j + k):\;a_0 ,a_1
\in Z} \right\} $$ which is a subset of quaternion integers. The
commutative property of multiplication holds over $ H(K_1 ) $.

\begin{theorem}For every odd, rational prime $p \in \mathcal{N}$, there exists a prime
$\pi  \in H(\mathcal{Z})$, such that $N(\pi ) = p = \pi \overline
\pi $. In particular, $p$  is not prime in $ H(Z) $ \cite{5}.
\end{theorem}

\begin{corollary}$\pi  \in H(\mathcal{Z})$ is prime in
$H(\mathcal{Z})$ if and only if $ N(\pi )$ is prime in
$\mathcal{Z}$ \cite{5}.
\end{corollary}

\begin{theorem}If  $a$ and $b$  are relatively prime integers then
$ H(K_1 )/\left\langle {a + b(i + j + k)} \right\rangle $ is
isomorphic to $ Z_{a^2  + 3b^2 } $ \cite{3,5,7}.
\end{theorem}

\section{Cyclic Codes over Quaternion Integer Rings }

Let $ H(K_1 )_{\pi ^k }$ be the residue class of $ H(K_1 )_\pi$
modulo $\pi^k$, where $k$ is any positive integer and $\pi$ is a
prime quaternion integer. According to the modulo function $\mu
:\mathcal{Z}_{P^k }  \to H(K_1 )_{\pi ^k }$ defined by
\begin{equation} \label{eq:1}  g \to g - [\frac{{g.\overline \pi
}}{{\pi .\overline \pi  }}]\pi \;(\bmod \pi ^k )\end{equation} $
H(K_1 )_{\pi ^k }$ is isomorphic to $Z_{p^k }$, where $p = \pi
\overline \pi$ and $p$ is an odd prime. A quaternion cyclic codes
$C$ of length $n$ is a linear code $C$ of length $n$ with property
$$(c_0 ,c_1 ,...,c_{n - 1} ) \in C \Rightarrow (c_{n - 1} ,c_0 ,c_1
,...,c_{n - 2} ) \in C.$$ In this case, we have a bijective
\begin{equation} \label{eq:2} \begin{array}{*{20}c}
   {H(K_1 )_{_{\pi ^k } }^n  \to H(K_1 )_{\pi ^k } [x]/(x^n  - 1)}  \\
   {(c_0 ,c_1 ,...,c_{n - 1} )\quad \;\quad  \mapsto c_0  + c_1 x +  \cdots  + c_{n - 1} x^{n - 1}  + (x^n  - 1)}  \\
\end{array}\end{equation}
To put it simply, we write $ c_0  + c_1 x +  \cdots  + c_{n - 1}
x^{n - 1}$ for $ c_0  + c_1 x +  \cdots  + c_{n - 1} x^{n - 1}  +
(x^n - 1)$. A nonempty set of $ H(K_1 )_{_{\pi ^k } }^n$ is a
$H(K_1 )_{\pi ^k }$-cyclic code if and only if its image under (2)
is an ideal of $ H(K_1 )_{\pi ^k } [x]/(x^n  - 1)$. More
information on cyclic codes can be found in \cite{6}.

\begin{definition}Let $ \alpha ,\beta  \in H(K_1 )_\pi
$ and $ \gamma  = \beta  - \alpha  = a + b(i + j + k)\;(\bmod \pi
)$, where $\pi$  is a prime quaternion integer. Let the quaternion
Mannheim weight of $\gamma$  be defined as $$ w_{QM} (\gamma ) =
\left| a \right| + 3\left| b \right|$$ the quaternion Mannheim
distance $d_{QM} $ between $\alpha$ and $\beta$ is defined as $$
d_{QM} (\alpha ,\beta ) = w_{QM} (\gamma ). \cite{3}$$
\end{definition}

\begin{proposition}Let $\pi=a+b(i+j+k)$ be a prime in the set $H(K_1 )$ and
let $p=a^2+3b^2$ be prime in $\mathcal{Z}$. If $g$ is a generator
of $H(K_1 )_{\pi ^2 }^ *$, then $g^{\phi (p^2 )/2}  \equiv  -
1\;(\bmod \,\pi ^2 )$.
\end{proposition}

\begin{proof}
If $N(\pi)$ is a prime integer in $\mathcal{Z}$, then the complete
part and the coefficient of the vector part of $\pi^2$ are
relatively integer. So, $\mathcal{Z}_{p^2}$ is isomorphic to
$H(K_1 )_{\pi ^2 }$ (See Theorem 2). If $g$ is a generator of
$H(K_1 )_{\pi ^2 }^ *$, then $g,\,g^2 ,\,...,\,g^{\phi (p^2 )}$
constitute a reduced residue system modulo $\pi ^2$ in $H(K_1
)_{\pi ^2 }$. Therefore, there is a positive integer $k$ as $g^k
\equiv  - 1\,(\bmod \,\pi ^2 )$, where $1 \le k \le \phi (p^2 )$.
Hence, we can infer $g^{2k}  \equiv 1\,(\bmod \,\pi ^2 )$. Since
$\left. {\phi (p^2 )} \right|2k$ and $2 \le 2k \le 2\phi (p^2 )$,
we obtain $\phi (p^2 )=k$ or $\phi (p^2 )=2k$. If $\phi (p^2 )$
was equal to $k$, we should have $\left. {\pi ^2 } \right|2$, but
this would contradict the fact that $ N(\pi ^2 ) > 2$.
\end{proof}

\begin{proposition} Let $\pi _k  = a_k  + b_k (i + j + k)$ be
distinct primes in $H(K_1 )$ and let $p_k  = a_k^2  + 3b_k^2$ be
distinct primes in $\mathcal{Z}$, where $k = 1,2,...,m$. If $g$ is
a generator of $H(K_1 )_{\pi ^k }^ *$, then $g^{\phi (p^k )/2}
\equiv  - 1\;(\bmod \,\pi ^k )$.
\end{proposition}

\begin{proof}This is certain from Proposition 1.
\end{proof}

\begin{theorem} Let $\pi  = a + b(i + j + k)$ be a prime in $H(K_1
)$ and let $p = a^2  + 3b^2$ be a prime in $\mathcal{Z}$, where
$a,b \in \mathcal{Z}$. Then, cyclic codes whose lengths are $\phi
(p^2 )/2$ are obtained.
\end{theorem}

\begin{proof} $H(K_1 )_{\pi ^2 }^ *$ has a generator since $\mathcal{Z}_{p^2 }  \cong H(K_1 )_{\pi ^2
}$. Let the generator be $g$. Then we get $g^{\phi (p^2 )}  = 1$
and $g^{\phi (p^2 )/2}  =  - 1 $. Hence, we can write $$x^{{{\phi
(p^2 )} \mathord{\left/
 {\vphantom {{\phi (p^2 )} 2}} \right.
 \kern-\nulldelimiterspace} 2}}  + 1 = (x - g)Q(x)\;(\bmod \,\pi ^2 )\;({\rm for}\,x = g{\rm
 )}.$$ In this situation, $(x - g)$ is an ideal of ${{H(K_1 )_{\pi ^2 } [x]} \mathord{\left/
 {\vphantom {{H(K_1 )_{\pi ^2 } [x]} {\left\langle {x^{{{\phi (p^2 )} \mathord{\left/
 {\vphantom {{\phi (p^2 )} 2}} \right.
 \kern-\nulldelimiterspace} 2}}  + 1} \right\rangle }}} \right.
 \kern-\nulldelimiterspace} {\left\langle {x^{{{\phi (p^2 )} \mathord{\left/
 {\vphantom {{\phi (p^2 )} 2}} \right.
 \kern-\nulldelimiterspace} 2}}  + 1} \right\rangle }}$, i.e., it generates a cyclic code.
\end{proof}

If the generator polynomial is taken as a monic polynomial, all
components of any row of the generator matrix do not consist of
zero divisors. Therefore, these codes are free $H(K_1 )_{\pi ^2 }$
modules.

\begin{proposition}Let $\pi _1  = a + b(i + j + k)$,
$\pi _2  = c + d(i + j + k)$ be primes in $H(K_1 )$ and let $p_1 =
a^2  + 3b^2$, $p_2  = c^2  + 3d^2$ be primes in $\mathcal{Z}$.
Then, there are two elements of $H(K_1 )_{\pi _1 \pi _2 }^ *$ such
that $e^{\phi (p_2 )}  \equiv 1\;(\bmod \pi _1 \pi _2 )$ and
$f^{\phi (p_1 )}  \equiv 1\;(\bmod \pi _1 \pi _2 )$.
\end{proposition}

\begin{proof} Since  $p_1$ and $p_2$  are relatively prime
integers in $\mathcal{Z}$, $\pi _1$ and $\pi _2$ are relatively
primes in $H(K_1 )$. Using the basic algebraic knowledge and the
function (1), we get $\mathcal{Z}_{p_1 }  \cong H(K_1 )_{\pi _1
}$, $\mathcal{Z}_{p_2 }  \cong H(K_1 )_{\pi _2 }$ and
$\mathcal{Z}_{p_1 p_2 }  \cong H(K_1 )_{\pi _1 \pi _2 }$.
Moreover, we obtain as follows: $$H(K_1 )_{\pi _1 \pi _2 }^ * (\pi
_1 ) \cong \mathcal{Z}_{p_1 p_2 }^ *  (p_1 ) \cong
\mathcal{Z}_{p_2 }^ * \cong H(K_1 )_{\pi _2 }^ *  ,$$ $$H(K_1
)_{\pi _1 \pi _2 }^ * (\pi _2 ) \cong \mathcal{Z}_{p_1 p_2 }^ *
(p_2 ) \cong \mathcal{Z}_{p_1 }^ *   \cong H(K_1 )_{\pi _1 }^ *
.$$ Since $\pi _2$ is a prime quaternion integer, $H(K_1 )_{\pi _2
}^ *$ is a cyclic group. Therefore, $H(K_1 )_{\pi _2 }^ *$ has a
generator. So, $H(K_1 )_{\pi _1 \pi _2 }^ *  (\pi _1 )$ has a
generator, either. Let the generator be $e$. Then $e^{\phi (p_2 )}
\equiv 1\;(\bmod \pi _1 \pi _2 )$. In the same way, $H(K_1 )_{\pi
_1 \pi _2 }^ *  (\pi _2 )$ has a generator. Suppose that $f$  is
the generator of $H(K_1 )_{\pi _1 \pi _2 }^ *  (\pi _2 )$. Then
$f^{\phi (p_1 )}  \equiv 1\;(\bmod \pi _1 \pi _2 )$.
\end{proof}

\begin{proposition} Let $\pi _k  = a_k  + b_k (i + j + k)$ be
a prime in $H(K_1 )$ and let $p_k  = a_k^2  + 3b_k^2$ be distinct
odd primes in $\mathcal{Z}$. Then, there is an element $e_k$ of
$H(K_1 )_{\pi _1 \pi _2 ...\pi _k }^ *$ such that $e_k^{\phi (p_k
)}  \equiv 1\;(\bmod \pi _1 \pi _2 ...\pi _k )$, $k = 1,2,...,m$.
\end{proposition}

\begin{proof} This is clear from Proposition 3.
\end{proof}

\begin{theorem} Let $\pi _1  = a + b(i + j + k),\,\pi _2  = c + d(i + j +
k)$ be primes in $H(K_1 )$ and let $p_1  = a^2  + 3b^2 ,\,p_2 =
c^2  + 3d^2$ be odd primes in $\mathcal{Z}$. Then, we can always
write cyclic codes of length $\phi (p_1 )$ and $\phi (p_2 )$ over
$H(K_1 )_{\pi _1 \pi _2 }$.Moreover, the generator polynomials of
these codes are first degree monic polynomials. Therefore, these
codes are free $H(K_1 )_{\pi _1 \pi _2 }$ module.
\end{theorem}

\begin{proof} From Proposition 3, we can find an element of $H(K_1 )_{\pi _1 \pi _2
}$ such that $e^{\phi (p_2 )}  \equiv 1\;(\bmod \pi _1 \pi _2 )$.
Thus, we factorize the polynomial $x^{\phi (p_2 )}  - 1$ over
$H(K_1 )_{\pi _1 \pi _2 }$ as $x^{\phi (p_2 )}  - 1 = (x -
e)D(x)(\bmod \pi _1 \pi _2 )$. If we take the generator polynomial
as $g(x)=x-e$ , then the generator polynomial $g(x)$ forms the
generator matrix whose all components of any rows do not consist
of zero divisors.
\end{proof}

We now consider a simple example with regard to Theorem 3.

\begin{example} Let $\pi$ be $2+i+j+k$. The polynomial $x^{21}
+1$ factors over $H(K_1 )_{\pi ^2 } $ as $x^{21}  + 1 = (x -
\alpha ).(x^{20}  + \alpha x^{19}  + \alpha ^2 x^{18}  + \alpha ^3
x^{17}  + ... + \alpha ^{19} x + \alpha ^{20} )$, where $\alpha  =
1 - i - j - k$. The powers of $\alpha$ are shown in Table I. If we
choose the generator polynomial as $g(x) = x - \alpha $, then the
generator matrix is as follows: $$G = \left(
{\begin{array}{*{20}c}
   { - \alpha } & 1 & 0 & 0 &  \cdots  & 0  \\
   0 & { - \alpha } & 1 & 0 &  \cdots  & 0  \\
   0 & 0 & { - \alpha } & 1 &  \ddots  & 0  \\
    \vdots  &  \vdots  & {} &  \ddots  &  \ddots  & {}  \\
   0 & 0 &  \cdots  & 0 & { - \alpha } & 1  \\
\end{array}} \right)_{20{\rm x21}} .$$ The code $C$ generated by the generator
matrix $G$ can correct one error having quaternion Mannheim weight
of one.
\end{example}

\begin{center}
{\scriptsize {Table I: Powers of the element $\alpha=1-i-j-k$
which is root of $x^3+1$.}} {\small \centering}
\begin{tabular}{|c|c|c|c|c|c|c|c|}
  \hline
  $s$ & $\alpha^s$ & $s$ & $\alpha^s$ & $s$ & $\alpha^s$ & $s$ & $\alpha^s$ \\
  \hline
  0   &        1     &  6 & $3+i+j+k$        & 12& $4-2i-2j-2k$& 18 & $-6$ \\
  1   & $1-i-j-k$    &  7 & $-6-i-j-k$       & 13& $2i+2j+2k$  & 19 & $5+i+j+k$ \\
  2   & $-1+2i+2j+2k$&  8 & $2$              & 14& $5-2i-2j-2k$& 20 & $-3+i+j+k$\\
  3   & $4-i-j-k$    &  9 & $2-2i-2j-2k$     & 15& $1+i+j+k$   & 21 & $-1$      \\
  4   & $2-i-j-k$    &  10 &$-3$             & 16& 4           & 22 & $ - \alpha  =  - 1 + i + j + k$
          \\
  5   & $i+j+k$      &  11& $-4-i-j-k$       & 17& 5           & 23 & $ - \alpha ^2  = 1 - 2i - 2j - 2k$
          \\
  \hline
\end{tabular}
\end{center}


\begin{thebibliography}{9}
\bibitem{1} Huber K., "Codes Over Gaussian Integers" IEEE Trans. Inform.Theory, vol. 40, pp. 207-216, Jan. 1994.
\bibitem{2} M. \"{O}zen and M. G\"{u}zeltepe, "Cyclic Codes over Some Finite
Rings" (Submitted 2009).
\bibitem{3} M. \"{O}zen and M. G\"{u}zeltepe, "Codes over Quaternion
Integers" (Submitted 2009).
\bibitem{4} C. Martinez, E. Stafford, R. Beivide, E. Gabidulin, "Perfect Codes over Lipschitz Integers" IEEE Int. Symposium, ISIT 2007.
\bibitem{5} G. Davidoff, P. Sarnak, A. Valette, "Elementary Number Theory, Group Theory, Ramanujan Graphs", Cambridge University Pres, 2003.
\bibitem{6} F. J. Macwilliams and N. J. SLOANE, "The Theory of Error Correcting  Codes", North Holland Pub. Co., 1977 .
\bibitem{7} G. Dresden and W. M. Dymacek, "Finding Factors of Factor Rings over the Gaussian Integers" The Mathematical Association of America, Monthly Aug-Sep. 2005.
\bibitem{8} I. Ziven, H.S. Zuckerman and H.L. Montgomery, "An Introduction to the Number Theory" John Wiley and Sons, Inc., 1991.
\end{thebibliography}
\end{document}